\documentclass[11pt,a4paper]{article}

\usepackage{amsmath,amsthm,amssymb,amsbsy,graphicx,graphics,dsfont} % ,pdfpages}
\usepackage[margin=2.5cm]{geometry}
\usepackage{url}
\usepackage{booktabs}
\usepackage[T1]{fontenc}
\usepackage[]{times}
\usepackage[latin1]{inputenc}
\usepackage{bm}
\usepackage{amsfonts}
\usepackage{amstext}
\usepackage{authblk}
\usepackage{subfigure}
\RequirePackage[colorlinks,citecolor=blue,urlcolor=blue]{hyperref}

\theoremstyle{definition}
\newtheorem{theorem}{Theorem}[section]
\newtheorem{definition}[theorem]{Definition}
\newtheorem{example}[theorem]{Example}

\newtheorem{corollary}[theorem]{Corollary}
\newtheorem{lemma}[theorem]{Lemma}

\usepackage{natbib}
\title{Ensemble Copula Coupling as a\\ Multivariate
Discrete Copula Approach}

\author{Roman Schefzik}

\affil{{\em Institute for Applied Mathematics, Heidelberg University}\\ {\em Im Neuenheimer Feld 294, 69120 Heidelberg}\\ {\em e-mail:} \texttt{r.schefzik@uni-heidelberg.de}}

\begin{document}

\maketitle

\begin{abstract}
\noindent In probability and statistics, copulas play important roles theoretically as well as to address a wide 
range of problems in various application areas.\\
In this paper, we introduce the concept of multivariate 
discrete copulas, discuss their equivalence to stochastic arrays, and provide a multivariate discrete version of Sklar's 
theorem.\\
These results provide the theoretical frame for the ensemble copula coupling approach proposed by \cite{Schefzik&2013} for the multivariate statistical postprocessing of weather forecasts made 
by ensemble systems.
\end{abstract}

\section{Introduction}\label{intro}

\noindent Originally introduced by \cite{sklar}, copulas play an important role in probability and statistics whenever modeling of stochastic dependence is required. Roughly speaking, copulas are functions that link multivariate distribution functions to their univariate marginal distribution functions, as is manifested in the famous Sklar's theorem \citep{sklar}. The field of copulas has been developing rapidly over the last decades, and copulas have been applied to a wide range of problems in various areas such as climatology, meteorology and hydrology \citep{moeller,GenFav2007,SchoelzelFriederichs2008,zhang} or econometrics, insurance  and mathematical finance \citep{Cher2004,Embrechts&2003,pfeifernes,gengenbou}. However, copulas are also of immense theoretical interest, due to their appealing mathematical properties. For a general overview of the mathematical theory of copulas, we refer to the textbooks by \cite{joe} and \cite{Nelsen2006}, as well as to the survey paper by \cite{Sempi2011}.\\
A special type of copulas are the so-called discrete copulas, whose properties have been studied by \cite{koles}, \cite{Mayor&2005}, \cite{Mayor&2007} and \cite{Mesiar2005} in recent years. However, the discussion in the papers mentioned above focuses on the bivariate case, and it is natural to search for a treatment of the general multivariate situation. In what follows, we generalize both the notion of discrete copulas and the most important results in this context to the multivariate case, and show to what extent they build the theoretical frame of the ensemble copula coupling (ECC) approach recently proposed by \cite{Schefzik&2013}. ECC is a multivariate statistical postprocessing tool for ensemble weather forecasts in meteorology, which turns out to be based on the theoretical framework discussed here.% In this regard, the paper at hand can be considered as kind of an addendum to the work of \cite{Schefzik&2013}.
\\
The remainder of this paper is organized as follows. In Section \ref{basics}, we introduce the multivariate discrete copula concept. We then point out the connection between multivariate discrete copulas and stochastic arrays \citep{Csima1970,MarchiTaraz1979} in Section \ref{dcstar} and continue with the formulation of a multivariate discrete version of Sklar's theorem in Section \ref{mdcsklar}. Eventually, Section \ref{applecc} deals with the ensemble copula coupling (ECC) approach \citep{Schefzik&2013} and its relationships to the presented results. 

\section{Multivariate discrete copulas}\label{basics}

\noindent First, we transfer the notion of bivariate discrete copulas introduced by \cite{koles} to the general multivariate case. Although our new class of copulas turns out to be a special Fr\'{e}chet class \citep{joe}, we nevertheless give all relevant definitions in detail, as they provide the basic concepts required in the subsequent sections.\newline
\newline
\noindent Let $I_{M}:=\left\{0,\frac{1}{M},\frac{2}{M},\ldots,\frac{M-1}{M},1\right\}$, where $M \in \mathbb{N}$, and $\overline{\mathbb{R}}=\mathbb{R} \cup \{-\infty,\infty\}$. 
\begin{definition}\label{mdc}
A function $D: I_{M}^{L} \rightarrow [0,1]$ is called a discrete copula on $I_{M}^{L}:=\underbrace{I_{M} \times \cdots \times I_{M}}_{L \operatorname{\, times}}$ if it satisfies the following conditions:
\begin{itemize}
\item[(D1)]{$D$ is grounded in the sense that $D\left(\frac{i_1}{M},\ldots,\frac{i_L}{M}\right)=0$
% if at least one component of $(\frac{i_{1}}{M},\ldots,\frac{i_{L}}{M}) \in I_{M}^{L}$ is equal to zero, that is, 
if  $i_{\ell}=0$ for at least one $\ell \in \{1,\ldots,L\}$.}
\item[(D2)]{$D(1,\ldots,1,\frac{i_{\ell}}{M},1,\ldots,1)=\frac{i_{\ell}}{M}$ for all $\ell \in \{1,\ldots,L\}$.}
\item[(D3)]{$D$ is $L$-increasing in the sense that $\Delta_{i_{L}-1}^{i_{L}} \cdots \Delta_{i_{1}-1}^{i_{1}} D (\frac{j_{1}}{M},\ldots,\frac{j_{L}}{M}) \geq 0$ for all $i_{\ell} \in \{1,\ldots,M\}$, $\ell \in \{1,\ldots,L\}$ and $j_{\ell} \in \{0,\ldots,M\}$, where
\begin{eqnarray*}
\Delta_{i_{\ell}-1}^{i_{\ell}} D \left(\frac{j_{1}}{M},\ldots,\frac{j_{L}}{M}\right)&:=&D\left(\frac{j_{1}}{M},\ldots,\frac{j_{\ell-1}}{M},\frac{i_{\ell}}{M},\frac{j_{\ell+1}}{M},\ldots,\frac{j_{L}}{M}\right)\\&&-D\left(\frac{j_{1}}{M},\ldots,\frac{j_{\ell-1}}{M},\frac{i_{\ell}-1}{M},\frac{j_{\ell+1}}{M},\ldots,\frac{j_{L}}{M}\right).
\end{eqnarray*}
}
\end{itemize}
\end{definition}
\begin{definition}\label{irrmdc}
A discrete copula $D:  I_{M}^{L} \rightarrow [0,1]$ is called irreducible if it has minimal range, that is, $\mbox{Ran}(D)=I_{M}$.
\end{definition}
\begin{definition}\label{mdsubcop}
A function $D^{\ast}: J_{M}^{(1)} \times \cdots \times J_{M}^{(L)} \rightarrow [0,1]$ with $\{0,1\} \subset J_{M}^{(1)},\ldots, J_{M}^{(L)} \subset I_{M} $ is called a discrete subcopula if it satisfies the following conditions:
\begin{itemize}
\item[(S1)]{$D^{\ast}\left(\frac{i_1}{M},\ldots,\frac{i_L}{M}\right)=0$
% if at least one component of $(\frac{i_{1}}{M},...,\frac{i_{L}}{M}) \in J_{M}^{(1)} \times \cdots \times J_{M}^{(L)} $ is equal to zero, that is,
if $i_{\ell}=0$ for at least one $\ell \in \{1,\ldots,L\}$. }
\item[(S2)]{$D^{\ast}(1,\ldots,1,\frac{i_{\ell}}{M},1,\ldots,1)=\frac{i_{\ell}}{M}$ for all $\frac{i_{\ell}}{M} \in J_{M}^{(\ell)}$.}
\item[(S3)]{$\Delta_{i_{L}}^{j_{L}} \cdots \Delta_{i_{1}}^{j_{1}} D^{\ast} \left(\frac{k_{1}}{M},\ldots,\frac{k_{L}}{M}\right) \geq 0$ for all $(\frac{i_{1}}{M},\ldots,\frac{i_{L}}{M}), (\frac{j_{1}}{M},\ldots,\frac{j_{L}}{M}) \in J_{M}^{(1)} \times \cdots \times J_{M}^{(L)}$ such that $i_{\ell} \leq j_{\ell}$ for all $\ell \in \{1,\ldots,L\}$, where
\begin{eqnarray*}
\Delta_{i_{\ell}}^{j_{\ell}} D^{\ast} \left(\frac{k_{1}}{M},\ldots,\frac{k_{L}}{M}\right)&:=&D^{\ast}\left(\frac{k_{1}}{M},\ldots,\frac{k_{\ell-1}}{M},\frac{j_{\ell}}{M},\frac{k_{\ell+1}}{M},\ldots,\frac{k_{L}}{M}\right)\\&&-D^{\ast}\left(\frac{k_{1}}{M},\ldots,\frac{k_{\ell-1}}{M},\frac{i_{\ell}}{M},\frac{k_{\ell+1}}{M},\ldots,\frac{k_{L}}{M}\right).
\end{eqnarray*}
}
\end{itemize}
\end{definition}
\noindent The definition of discrete (sub)copulas can be generalized in the following way: A discrete copula need not necessarily have domain $I_{M}^{L}$, but can generally be defined on $I_{M_{1}} \times \cdots \times I_{M_{L}}$, where $M_{1},\ldots,M_{L} \in \mathbb{N}$ might take distinct values. Then, the axioms (D1), (D2) and (D3) apply analogously to this case. Similarly, discrete subcopulas can generally be defined on $J_{M_{1}}^{(1)} \times \cdots \times J_{M_{L}}^{(L)}$ for possibly distinct numbers $M_{1},\ldots,M_{L} \in \mathbb{N}$, taking account of the conditions in Definition \ref{mdsubcop}.\\
However, for convenience and in view of the application in Section \ref{applecc}, we confine ourselves to the case of $M:=M_{1}=\cdots=M_{L}$ as in the above Definitions \ref{mdc} to \ref{mdsubcop}.
\newline
\noindent Following Chapter 3 in \cite{joe}, a multivariate discrete copula can be interpreted as a multivariate distribution in the Fr\'{e}chet class $\mathcal{F}(F_{\mathcal{U}(I_M)},\ldots,F_{\mathcal{U}(I_M)})$, where $F_{\mathcal{U}(I_M)}$ is the cumulative distribution function (cdf) of a uniformly distributed random variable on $I_M$.
\newline
\noindent Let us now give first explicit examples for multivariate discrete copulas.
\begin{example}\label{pim}
Let $i_1,\ldots,i_M \in \{0,1,\ldots,M\}$.
\begin{enumerate}
\item[(a)]{$\Pi\left(\frac{i_1}{M},\ldots,\frac{i_L}{M}\right):=\prod\limits_{\ell=1}^{L}\frac{i_{\ell}}{M}$ is a discrete copula, the so-called product or independence copula.}
\item[(b)]{${\cal{M}}\left(\frac{i_1}{M},\ldots,\frac{i_L}{M}\right):=\mbox{min}\left\{\frac{i_{1}}{M},\ldots,\frac{i_{L}}{M}\right\}$ is an irreducible discrete copula.}
\end{enumerate}
\end{example}
\noindent Note that $\Pi$ and ${\cal{M}}$ are indeed multivariate discrete copulas because they represent the restrictions of two well-known standard copulas defined on $[0,1]^{L}$ to the discrete set $I_{M}^{L}$.
\begin{example}\label{emcop}
Another example for an irreducible discrete copula  
is given by the so-called empirical copula, which will be very important with respect to the ECC approach, see Section \ref{applecc}. Let  
${\cal{S}}=\{(x_{1}^{1},\ldots,x_{1}^{L}),\ldots,(x_{M}^{1},\ldots,x_{M}^{L})\}$, where  
$x_{m}^{\ell} \in  \mathbb{R}$ for all $m \in \{1,\ldots,M\}$ and $\ell  
\in \{1,\ldots,L\}$ with $x_{m}^{1} \neq x_{n}^{1},\ldots,x_{m}^{L} \neq  
x_{n}^{L}$ for $m \in \{1,\ldots,M\}$, $n \in \{1,\ldots,M\}$, $m \neq n$. That is, we assume for simplicity that there are no ties among the respective samples. Moreover, let  
$x_{(1)}^{1}< \ldots < x_{(M)}^{1},\ldots,x_{(1)}^{L}<\ldots<x_{(M)}^{L}$ be the  
(marginal) order statistics of the collections  
$\{x_{1}^{1},\ldots,x_{M}^{1}\},\ldots,\{x_{1}^{L},\ldots,x_{M}^{L}\}$,  
respectively.\\
Then, the empirical copula $E_{M}:I_{M}^{L} \rightarrow I_{M}$ defined  
from ${\cal{S}}$ is given by
\begin{align*}
E_{M}\left(\frac{i_{1}}{M},\ldots,\frac{i_{L}}{M}\right)&:=  
\begin{cases} 0&\mbox{if }i_{\ell}=0 \mbox{\,\,for at least one } \ell \in  
\{1,\ldots,L\} \\
\frac{\#\{(x_{m}^{1},\ldots,x_{m}^{L}) \in {\cal{S}} | x_{m}^{1} \leq  
x_{(i_{1})}^{1},\ldots,x_{m}^{L} \leq x_{(i_{L})}^{L} \}}{M}&\mbox{if\,  
}i_{\ell} \in \{1,\ldots,M\} \mbox{\,\,for all } \ell \in \{1,\ldots,L\}
\end{cases}
\end{align*}
or, equivalently,
\begin{eqnarray*}
E_{M}\left(\frac{i_{1}}{M},\ldots,\frac{i_{L}}{M}\right) &:=&\frac{1}{M} \sum\limits_{m=1}^{M}\mathds{1}_{\{\operatorname{rk}(x_{m}^{1}) \leq i_1,\ldots,\operatorname{rk}(x_{m}^{L})\leq i_L\}} = \frac{1}{M}  
\sum\limits_{m=1}^{M} \prod\limits_{\ell=1}^{L}  
\mathds{1}_{\left\{\operatorname{rk}(x_{m}^{\ell}) \leq i_{\ell}\right\}},
\end{eqnarray*}
where $\mbox{rk}(x_{m}^{\ell})$ denotes the rank of $x_{m}^{\ell}$ in  
$\{x_{1}^{\ell},\ldots,x_{M}^{\ell}\}$ for $\ell \in \{1,\ldots,L\}$, and $ i_{1},\ldots,i_{L} \in \{0,1,\ldots,M\}$, compare \cite{deheuvels}.
\\
Obviously, the empirical copula is an irreducible discrete copula.  
Conversely, any irreducible discrete copula is the empirical copula of  
some set ${\cal{S}}$, as discussed in Example \ref{exequiv} (c) in  
Section \ref{dcstar}.
\end{example}

\section{A characterization of multivariate discrete copulas via stochastic arrays}\label{dcstar}

\noindent According to \cite{koles} and \cite{Mayor&2005}, there is a one-to-one correspondence between discrete copulas and bistochastic matrices in the bivariate case. We now formulate a similar characterization for multivariate discrete copulas. To this end, the notion of stochastic arrays \citep{Csima1970,MarchiTaraz1979} turns out to be very useful.
\begin{definition}\label{stocharr}
An array $A=(a_{i_{1} \ldots i_{L}})_{i_{1},\ldots,i_{L}=1}^{M}$ is called an $L$-dimensional stochastic array (or an $L$-stochastic matrix) of degree $L-1$ if 
\begin{itemize}
\item[(A1)]{$a_{i_{1} \ldots i_{L}} \geq 0$ for all $i_{1},\ldots,i_{L} \in \{1,\ldots,M\}$}
\item[(A2)]{$\sum\limits_{i_{1},\ldots,i_{\ell-1},i_{\ell+1},\ldots,i_{L}}a_{i_{1} \ldots i_{L}}:=\sum\limits_{i_{1}=1}^{M} \cdots \sum\limits_{i_{\ell-1}=1}^{M} \sum\limits_{i_{\ell+1}=1}^{M} \cdots \sum\limits_{i_{L}=1}^{M} a_{i_{1} \ldots i_{\ell-1} i_{\ell} i_{\ell+1} \ldots i_{L}}=1$ \\ for $i_{\ell} \in \{1,\ldots,M\}$, where $\ell \in \{1,\ldots,L\}$, and the summation is over $i_{1},\ldots,i_{\ell-1},i_{\ell+1},\ldots,i_{L} \in \{1,\ldots,M\}$.}
\end{itemize}
\noindent As a special case, an $L$-dimensional stochastic array $A$ is called an $L$-dimensional permutation array (or an $L$-permutation matrix) if the entries of $A$ only take the values 0 and 1, that is, $a_{i_{1} \ldots i_{L}} \in \{0,1\}$ for all $i_{1},\ldots,i_{L} \in \{1,\ldots,M\}$.
\end{definition}
\begin{theorem}\label{equiv}
Let $D:I_{M}^{L} \rightarrow [0,1]$. Then, the following statements are equivalent:
\begin{enumerate}
\item[(1)]{$D$ is a discrete copula.}
\item[(2)]{There exists an $L$-dimensional stochastic array $A=(a_{i_{1} \ldots i_{L}})_{i_{1},\ldots,i_{L}=1}^{M}$ such that
\begin{equation}\label{stochm}
D \left(\frac{i_{1}}{M},\ldots,\frac{i_{L}}{M}\right)=\frac{1}{M} \sum\limits_{j_{1}=1}^{i_{1}} \cdots \sum\limits_{j_{L}=1}^{i_{L}} a_{j_{1} \ldots j_{L}}
\end{equation}
for $i_{1},\ldots,i_{L} \in \{0,1,\ldots,M\}$.
}
\end{enumerate}
\end{theorem}
\begin{corollary}\label{corequiv}
$D$ is an irreducible discrete copula if and only if there is an $L$-dimensional permutation array $A=(a_{i_{1} \ldots i_{L}})_{i_{1},\ldots,i_{L}=1}^{M}$ such that $\eqref{stochm}$ holds for $i_{1},\ldots,i_{L} \in \{0,1,\ldots,M\}$.
\end{corollary}
\noindent The proof of Theorem \ref{equiv} basically consists in showing the validity of the axioms (A1), (A2), (D1), (D2) and (D3) in Definitions \ref{stocharr} and \ref{mdc}, respectively. This is on the one hand straightforward, but on the other hand rather tedious, involving several calculations of multiple sums. We omit a detailed proof and stress that Theorem \ref{equiv} can also be interpreted as a reformulation of the relation between the cdf and the probabiliy mass function (pmf) \citep{Xu1996} because the stochastic array in Definition \ref{stocharr} can be identified with $M$ times the pmf.\newline
Essentially, Theorem \ref{equiv} yields the following equivalences:
\begin{center}
discrete copula $\Leftrightarrow$ marginal distributions concentrated on $\left\{\frac{1}{M},\frac{2}{M},\ldots,1 \right\}$ $\Leftrightarrow$ probability masses on $\left\{\frac{1}{M},\frac{2}{M},\ldots,1 \right\}^L$ $\Leftrightarrow$ stochastic array.
\end{center}
In the situation of Corollary \ref{corequiv}, we have
\begin{center}
irreducible discrete copula $\Leftrightarrow$ empirical copula $\Leftrightarrow$ $M$ point masses of $\frac{1}{M}$ each $\Leftrightarrow$ permutation array $\Leftrightarrow$ Latin hypercube of order $M$ in $L$ dimensions \citep{Gupta1974}.
\end{center}
Illustrations of these equivalences are given in Section \ref{applecc}, where we discuss their relevance with respect to the ECC approach of \cite{Schefzik&2013}.
\begin{example}\label{exequiv}
$\hfill$
\begin{enumerate}
\item[(a)]{The discrete product copula $\Pi\left( \frac{i_{1}}{M},\ldots,\frac{i_{L}}{M}\right) = \prod\limits_{\ell=1}^{L}\frac{i_{\ell}}{M}$, where $(\frac{i_{1}}{M},\ldots,\frac{i_{L}}{M})\in I_{M}^{L}$, in Example \ref{pim} (a) corresponds to the $L$-dimensional stochastic array $A:=(a_{i_{1} \ldots i_{L}})_{i_{1},\ldots,i_{L}=1}^{M}=\left(\frac{1}{M^{L-1}}\right)_{i_{1},\ldots,i_{L}=1}^{M}$ whose entries are all equal to $\frac{1}{M^{L-1}}$. Indeed,
\begin{eqnarray*}
\frac{1}{M} \sum\limits_{j_{1}=1}^{i_{1}} \cdots \sum\limits_{j_{L}=1}^{i_{L}}\frac{1}{M^{L-1}} &=& \frac{1}{M} \sum\limits_{j_{1}=1}^{i_{1}} \cdots \sum\limits_{j_{L-1}=1}^{i_{L-1}}\frac{i_{L}}{M^{L-1}}
=\frac{1}{M} \cdot \frac{i_{1}\cdot \ldots \cdot i_{L-1}\cdot i_{L}}{M^{L-1}}\\  &=& \frac{i_{1} \cdot \ldots \cdot i_{L}}{M^{L}}
= \prod\limits_{\ell=1}^{L}\frac{i_{\ell}}{M}
= \Pi\left(\frac{i_{1}}{M},\ldots,\frac{i_{L}}{M}\right).
\end{eqnarray*}
}
\item[(b)]{The irreducible discrete copula ${\cal{M}}\left( \frac{i_{1}}{M},\ldots,\frac{i_{L}}{M}\right)= \mbox{min}\left\{ \frac{i_{1}}{M},\ldots,\frac{i_{L}}{M}\right\} $, where $(\frac{i_{1}}{M},\ldots,\frac{i_{L}}{M})\in I_{M}^{L}$, in Example \ref{pim} (b) corresponds to the $L$-dimensional identity stochastic array
\begin{equation*}
\mathbb{I}:=(a_{i_{1} \ldots i_{L}})_{i_{1},\ldots,i_{L}=1}^{M}, \mbox{\,\, where\,\,} a_{i_{1} \ldots i_{L}}=\begin{cases} 1&\mbox{if\,}i_{1}=\ldots=i_{L} \\ 0& \mbox{otherwise}.\end{cases}
\end{equation*}
Indeed, employing the definition and writing down the corresponding multiple sum explicitly yields
\begin{equation*}
\frac{1}{M}\sum\limits_{j_{1}=1}^{i_{1}} \cdots \sum\limits_{j_{L}=1}^{i_{L}}a_{j_{1} \ldots j_{L}} = \frac{1}{M} \cdot \mbox{min}\{i_{1},\ldots,i_{L}\}= \mbox{min}\left\{\frac{i_{1}}{M},\ldots,\frac{i_{L}}{M}\right\}
={\cal{M}}\left(\frac{i_{1}}{M},\ldots,\frac{i_{L}}{M}\right).
\end{equation*}
}
\item[(c)]{The empirical copula $E_{M}$ in Example \ref{emcop}, which is an  
irreducible discrete copula, corresponds to the $L$-dimensional  
permutation array $A=(a_{i_{1} \ldots i_{L}})_{i_{1},\ldots,i_{L}=1}^{M}$ with
\begin{equation*}
a_{i_{1} \ldots i_{L}}=\begin{cases} 1&\mbox{if } (x_{(i_{1})}^{1},\ldots,x_{(i_{L})}^{L}) \in  
{\cal{S}} \\
0&\mbox{if } (x_{(i_{1})}^{1},\ldots,x_{(i_{L})}^{L}) \notin {\cal{S}}.
\end{cases}
\end{equation*}
Conversely, for an irreducible discrete copula $D$ with associated  
$L$-dimensional permutation array\newline  
$A=(a_{i_{1} \ldots i_{L}})_{i_{1},\ldots,i_{L}=1}^{M}$, we consider the sets  
${\cal{X}}^{1}=\{x_{1}^{1}< \ldots <x_{M}^{1}\},\ldots,{\cal{X}}^{L}=\{x_{1}^{L}< \ldots <x_{M}^{L}\}$. Then, $D$ is the empirical copula of the set ${\cal{S}}=\{(x_{i_{1}}^{1},\ldots,x_{i_{L}}^{L}) | a_{i_{1} \ldots i_{L}}=1\}$.
}
\end{enumerate}
\end{example}

\section{A multivariate discrete version of Sklar's theorem}\label{mdcsklar}   

\noindent The most important result in the context of copulas is Sklar's theorem, see \cite{Nelsen2006} or \cite{sklar}. Our goal is now to prove a multivariate discrete version of Sklar's theorem, where the following extension lemma will play an essential role. A bivariate variant of this result has been shown by \cite{Mayor&2007}.
\begin{lemma} \label{extension} 
(Extension lemma) For each irreducible discrete subcopula $D^{\ast}: J_{M}^{(1)} \times \cdots \times J_{M}^{(L)} \rightarrow I_{M}$, there is an irreducible discrete copula $D: I_{M} \times \cdots \times I_{M} \rightarrow I_{M}$ such that
\begin{equation*}
D|_{J_{M}^{(1)} \times \cdots \times J_{M}^{(L)}} = D^{\ast},
\end{equation*} 
that is, the restriction of $D$ to $J_{M}^{(1)} \times \cdots \times J_{M}^{(L)}$ coincides with $D^{\ast}$.
\end{lemma}
\begin{proof} 
Let 
\begin{equation*}
J_{M}^{(\ell)} := \left\{0=\frac{a_{0}^{(\ell)}}{M} < \frac{a_{1}^{(\ell)}}{M} < \ldots < \frac{a_{r_{\ell}}^{(\ell)}}{M} < \frac{a_{r_{\ell}+1}^{(\ell)}}{M}=1 \right\}
\end{equation*}
for all $\ell \in \{1,\ldots,L\}$, with the corresponding equivalent sets
\begin{equation*}
K_{M}^{(\ell)} := \{0=a_{0}^{(\ell)} < a_{1}^{(\ell)} < \ldots < a_{r_{\ell}}^{(\ell)} <a_{r_{\ell}+1}^{(\ell)}=M \}.
\end{equation*}
To get an irreducible discrete extension copula $D$ of an irreducible discrete subcopula $D^{\ast}$, according to Theorem \ref{equiv}, it suffices to construct an $L$-dimensional permutation array $A$ such that each block specified by the points $(a_{s_{1}}^{(1)},a_{s_{2}}^{(2)},\ldots,a_{s_{L}}^{(L)})$ and $(a_{s_{1}+1}^{(1)},a_{s_{2}+1}^{(2)},\ldots,a_{s_{L}+1}^{(L)})$, which consists of the lines from 
$a_{s_{1}}^{(1)}+1$ to $a_{s_{1}+1}^{(1)}$, from $a_{s_{2}}^{(2)}+1$ to $a_{s_{2}+1}^{(2)}$, and so forth, up to the line from $a_{s_{L}}^{(L)}+1$ to $a_{s_{L}+1}^{(L)}$, contains a number of 1's equal to the volume
\begin{equation*}
M \cdot \left(\Delta_{a_{s_{L}}^{(L)}}^{a_{s_{L}+1}^{(L)}} \cdots \Delta_{a_{s_{1}}^{(1)}}^{a_{s_{1}+1}^{(1)}} D^{\ast} \left(\frac{j_{1}}{M},\ldots,\frac{j_{L}}{M}\right)\right),
\end{equation*}
where $s_{\ell} \in \{0,\ldots,r_{\ell}\}$ and $\ell \in \{1,\ldots,L\}$.\\
To show the existence of such a permutation array, let $k \in \{1,\ldots,L\}$ be fixed and consider the subarray specified by the lines $a_{s_{k}}^{(k)}+1$ and $a_{s_{k}+1}^{(k)}$ of the permutation array $A$. This subarray contains all the blocks determined by the points $(a_{s_{1}}^{(1)},\ldots,a_{s_{L}}^{(L)})$ and $(a_{s_{1}+1}^{(1)},\ldots,a_{s_{L}+1}^{(L)})$ for all $s_{\ell} \in \{0,\ldots,r_{\ell}\}$, where $\ell \in \{1,\ldots,L\} \setminus \{k\}$.\\
We need to show that the number $a_{s_{k}+1}^{(k)} - a_{s_{k}}^{(k)}$ of lines in this subarray is equal to the number of 1's corresponding to all those blocks. This indeed holds as
\begin{eqnarray*}
&&\sum\limits_{s_{1}=0}^{r_{1}} \cdots \sum\limits_{s_{k-1}=0}^{r_{k-1}}\sum\limits_{s_{k+1}=0}^{r_{k+1}} \cdots \sum\limits_{s_{L}=0}^{r_{L}} M \cdot \left(\Delta_{a_{s_{L}}^{(L)}}^{a_{s_{L}+1}^{(L)}} \cdots \Delta_{a_{s_{1}}^{(1)}}^{a_{s_{1}+1}^{(1)}} D^{\ast} \left(\frac{j_{1}}{M},\ldots,\frac{j_{L}}{M}\right)\right)\\
&=&M \cdot \underbrace{\sum\limits_{s_{1}=0}^{r_{1}} \cdots \sum\limits_{s_{k-1}=0}^{r_{k-1}}\sum\limits_{s_{k+1}=0}^{r_{k+1}} \cdots \sum\limits_{s_{L}=0}^{r_{L}}  \left(\Delta_{a_{s_{L}}^{(L)}}^{a_{s_{L}+1}^{(L)}} \cdots \Delta_{a_{s_{1}}^{(1)}}^{a_{s_{1}+1}^{(1)}} D^{\ast} \left(\frac{j_{1}}{M},\ldots,\frac{j_{L}}{M}\right)\right)}_{=:S} \\
&\stackrel{(\ast)}{=}& M \cdot \left( \Delta_{a_{s_{k}}^{(k)}}^{a_{s_{k}+1}^{(k)}} D^{\ast}\left(1,\ldots,1,\frac{j_{k}}{M},1,\ldots,1 \right)\right)\\
&=& M \cdot \left(D^{\ast} \left(1,\ldots,1,\frac{a_{s_{k}+1}^{(k)}}{M},1,\ldots,1 \right)- D^{\ast} \left(1,\ldots,1,\frac{a_{s_{k}}^{(k)}}{M},1,\ldots,1 \right) \right)\\
&\overset{\operatorname{Def. \, 2}}{\underset{\operatorname{(S2)}}{=}}& M \cdot \left(\frac{a_{s_{k}+1}^{(k)}}{M}-\frac{a_{s_{k}}^{(k)}}{M} \right)\\
&=&a_{s_{k}+1}^{(k)}-a_{s_{k}}^{(k)}.
\end{eqnarray*}
To see equality ($\ast$) in this connection, we let $\ell \in \{1,\ldots,L\} \setminus \{k\}$ be fixed and first consider the sum
\begin{equation*}
T:=\sum\limits_{s_\ell=0}^{r_\ell} \left(\Delta_{a_{s_{L}}^{(L)}}^{a_{s_{L}+1}^{(L)}} \cdots \Delta_{a_{s_{1}}^{(1)}}^{a_{s_{1}+1}^{(1)}} D^{\ast} \left(\frac{j_1}{M},\ldots,\frac{j_L}{M} \right) \right).
\end{equation*}
Writing down $T$ explicitly yields that all of the $(r_{\ell}+1) \cdot 2^L$ terms $D(\cdot,\ldots,\cdot)$ of $T$ cancel out except for those $2^L$ having a 0 or a 1 in the $\ell$-th component, which indeed occurs as $a_{0}^{(\nu)}=0$ and $a_{r_{\ell}+1}^{(\nu)}=M$ for $\nu \in \{1,\ldots,L\}$. According to property (S1) in Definition \ref{mdsubcop}, all the $2^{L-1}$ terms having a 0 in the $\ell$-th component vanish, and it remains
\begin{equation*}
T=\Delta_{a_{s_{L}}^{(L)}}^{a_{s_{L}+1}^{(L)}} \cdots \Delta_{a_{s_{\ell+1}}^{(\ell+1)}}^{a_{s_{\ell+1}+1}^{(\ell+1)}} \Delta_{a_{s_{\ell-1}}^{(\ell-1)}}^{a_{s_{\ell-1}+1}^{(\ell-1)}} \cdots \Delta_{a_{s_{1}}^{(1)}}^{a_{s_{1}+1}^{(1)}} D^{\ast} \left(\frac{j_1}{M},\ldots,\frac{j_{\ell-1}}{M},1,\frac{j_{\ell+1}}{M},\ldots,\frac{j_L}{M} \right).
\end{equation*}
By applying this iteratively and using again property (S1) in Definition \ref{mdsubcop}, we finally see that all but two of the terms $D(\cdot,\ldots,\cdot)$ of $S$ either vanish or cancel out, such that
\begin{equation*}
S=\Delta_{a_{s_{k}}^{(k)}}^{a_{s_{k}+1}^{(k)}} D^{\ast} \left(1,\ldots,1,\frac{j_k}{M},1,\ldots,1 \right),
\end{equation*}
and ($\ast$) is thus shown.\newline
Hence, we have proved that an irreducible discrete subcopula $D^{\ast}$ can be extended to an irreducible discrete copula $D$.
\end{proof}
\noindent Note that the extension proposed in Lemma \ref{extension} is in general not uniquely determined. In the case of non-uniqueness, there are a largest and a smallest discrete extension copula $D_{lar}$ and $D_{sm}$, respectively, in the sense that
\begin{equation*}
D_{lar}\left(\frac{i_{1}}{M},\ldots,\frac{i_{L}}{M}\right) \geq D\left(\frac{i_{1}}{M},\ldots,\frac{i_{L}}{M}\right) \geq D_{sm}\left(\frac{i_{1}}{M},\ldots,\frac{i_{L}}{M}\right) \mbox{\,\,for \,all\,\,} \left(\frac{i_{1}}{M},\ldots,\frac{i_{L}}{M}\right) \in I_{M}^{L}
\end{equation*}
for any other discrete extension copula $D$ of $D^{\ast}$.
\\
\\
With Lemma \ref{extension}, we are now ready to state and prove a multivariate discrete version of Sklar's theorem. For the bivariate case, such a result can be found in \cite{Mayor&2007}. 
\begin{theorem}\label{sklarmdc} 
(Sklar's theorem in the multivariate discrete case)
\begin{enumerate}
\item{Let $F_{1},\ldots,F_{L}$ be distribution functions with $\mbox{Ran}(F_{\ell}) \subseteq I_{M}$ for all $\ell \in \{1,\ldots,L\}$. If $D$ is an irreducible discrete copula on $I_{M}^{L}$, then
\begin{equation}\label{sklar}
H(x_{1},\ldots,x_{L})=D(F_{1}(x_{1}),\ldots,F_{L}(x_{L})) \, \mbox{\,\,for \,\,} (x_{1},\ldots,x_{L}) \in \overline{\mathbb{R}}^{L}
\end{equation}
is a joint distribution function with $\mbox{Ran}(H) \subseteq I_{M}$, having $F_{1},\ldots,F_{L}$ as marginal distribution functions.
}
\item{Conversely, if $H$ is a joint distribution function with marginal distribution functions $F_{1},\ldots,F_{L}$ and  $\mbox{Ran}(H) \subseteq I_{M}$, there exists an irreducible discrete copula $D$ on $I_{M}^{L}$ such that
\begin{equation*}
H(x_{1},\ldots,x_{L})=D(F_{1}(x_{1}),\ldots,F_{L}(x_{L})) \,\, \mbox{\,\,for \, \,} (x_{1},\ldots,x_{L}) \in \overline{\mathbb{R}}^{L}.
\end{equation*}
Furthermore, $D$ is uniquely determined if and only if $\mbox{Ran}(F_{\ell}) = I_{M}$ for all $\ell \in \{1,\ldots,L\}$.
}
\end{enumerate}
\end{theorem}
\begin{proof}$\hfill$
\begin{enumerate}
\item{This is just a special case of the common Sklar's theorem.}
\item{Let $H$ be a finite $L$-dimensional joint distribution function with $\mbox{Ran}(H)\subseteq I_{M}$ having one-dimensional marginal distribution functions $F_{1},\ldots,F_{L}$. Set
\begin{equation*}
J_{M}^{(\ell)} := \left\{\frac{i_{\ell}}{M} \in I_{M}\left| \frac{i_{\ell}}{M} \in \mbox{Ran}(F_{\ell})\right\} \supseteq \{0,1\}\right.
\end{equation*}
for all $\ell \in \{1,\ldots,L\}$ and define
\begin{equation*}
D^{\ast}:J_{M}^{(1)} \times \cdots \times J_{M}^{(L)} \rightarrow I_{M}, \, D^{\ast}\left(\frac{i_{1}}{M},\ldots,\frac{i_{L}}{M}\right):=H(x_{1},\ldots,x_{L}),
\end{equation*}
where $x_{\ell}$ satisfies $F_{\ell}(x_{\ell})=\frac{i_{\ell}}{M}$ for all $\ell \in \{1,\ldots,L\}$. We now show that $D^{\ast}$ is indeed an irreducible discrete subcopula. First, $\mbox{Ran}(H) \subseteq I_{M}$ by assumption, and $D^{\ast}$ is well-defined, due to the well-known fact that $H(x_1,\ldots,x_L)=H(y_1,\ldots,y_L)$ for points $(x_1,\ldots,x_L) \in \overline{\mathbb{R}}^L$ and $(y_1,\ldots,y_L) \in \overline{\mathbb{R}}^L$ such that $F(x_\ell)=F(y_\ell)$ for all $\ell \in \{1,\ldots,L\}$. Furthermore, the axioms (S1), (S2) and (S3) for discrete subcopulas in Definition \ref{mdsubcop} are fulfilled:
\begin{itemize}
\item[(S1)]{Let $i_{\ell}=0$ for an $\ell \in \{1,\ldots,L\}$. Then, $D^{\ast}(\frac{i_{1}}{M},\ldots,\frac{i_{\ell-1}}{M},0,\frac{i_{\ell+1}}{M},\ldots,\frac{i_{L}}{M})=H(x_{1},\ldots,x_{L})$ with $F_{\ell}(x_{\ell})=\frac{0}{M}=0$ and $F_{k}(x_{k})=\frac{i_{k}}{M}$ for all $k \in \{1,\ldots,L\} \setminus \{\ell\}$. However, $F_{\ell}(x_{\ell}) = H(\infty,\ldots,\infty,x_{\ell},\infty,\ldots,\infty) = 0$, and since $H$ is non-decreasing in each argument, we have $H(x_{1},\ldots,x_{\ell},\ldots,x_{L})=0$, and hence $D^{\ast}(\frac{i_{1}}{M},\ldots,\frac{i_{\ell-1}}{M},0,\frac{i_{\ell+1}}{M},\ldots,\frac{i_{L}}{M}) = 0$ for all $\frac{i_{k}}{M} \in J_{M}^{(k)}$, $k \in \{1,\ldots,L\} \setminus \{\ell\}$. Clearly, this is also true if $i_{\ell}=0$ for two or more $\ell \in \{1,\ldots,L\}$.}
\item[(S2)]{$D^{\ast}(1,\ldots,1,\frac{i_{\ell}}{M},1,\ldots,1)=H(x_{1},\ldots,x_{L})$ with $F_{\ell}(x_{\ell})=\frac{i_{\ell}}{M}$ and $F_{k}(x_{k})=1$ for all $k \in \{1,\ldots,L\} \setminus \{\ell\}$. Set $x_{k}:=\infty$ for all $k \in \{1,\ldots,L\} \setminus \{\ell\}$. Then,  $D^{\ast}(1,\ldots,1,\frac{i_{\ell}}{M},1,\ldots,1)=H(\infty,\ldots,\infty,x_{\ell},\infty,\ldots,\infty)=F_{\ell}(x_{\ell})=\frac{i_{\ell}}{M}$ for all $\frac{i_{\ell}}{M} \in J_{M}^{(\ell)}$.}
\item[(S3)]{To show that $D^{\ast}$ is $L$-increasing, we use the $L$-increasingness of $H$ as a finite distribution function and obtain
\begin{equation*}
\Delta_{i_{L}}^{j_{L}} \cdots \Delta_{i_{1}}^{j_{1}} D^{\ast}\left(\frac{k_{1}}{M},\ldots,\frac{k_{L}}{M}\right)=\Delta_{i_{L}}^{j_{L}} \cdots \Delta_{i_{1}}^{j_{1}} H(u_{1},\ldots,u_{L}) \geq 0
\end{equation*}
with $\frac{j_{\ell}}{M} \geq \frac{i_{\ell}}{M}$ for all $\frac{i_{\ell}}{M} \in J_{M}^{(\ell)}$ and $\frac{j_{\ell}}{M} \in J_{M}^{(\ell)}$, where $F_{\ell}(x_{\ell})=\frac{i_{\ell}}{M}$ and $F_{\ell}(y_{\ell})=\frac{j_{\ell}}{M}$ for all $x_{\ell} \in \overline{\mathbb{R}}$, $y_{\ell} \in \overline{\mathbb{R}}$ and $\ell \in \{1,\ldots,L\}$. This means that $D^{\ast}$ is $L$-increasing.
}
\end{itemize}
Thus, $D^{\ast}$ is indeed a subcopula.\\
According to Lemma \ref{extension}, $D^{\ast}$ can therefore be extended to a discrete copula $D$, which obviously satisfies
\begin{equation*}
D(F_{1}(x_{1}),\ldots,F_{L}(x_{L})) = D\left(\frac{i_{1}}{M},\ldots,\frac{i_{L}}{M}\right) = D^{\ast} \left(\frac{i_{1}}{M},\ldots,\frac{i_{L}}{M}\right) = H(x_{1},\ldots,x_{L})
\end{equation*}
for $(x_{1},\ldots,x_{L}) \in \overline{\mathbb{R}}^{L}$. Hence, $H(x_{1},\ldots,x_{L})=D(F_{1}(x_{1}),\ldots,F_{L}(x_{L}))$.\\
The last issue left to prove is that $D$ is uniquely determined if and only if $\mbox{Ran}(F_{\ell})=I_{M}$ for all $\ell \in \{1,\ldots,L\}$. Assume that $D^{\ast}:J_{M}^{(1)} \times \cdots \times J_{M}^{(L)} \rightarrow I_{M}$ can be extended in only a single way to a discrete copula $D$. Then, due to the unique extension, we have $D_{lar}=D=D_{sm}$, where $D_{lar}$ and $D_{sm}$ denote the largest and the smallest discrete extension copulas, respectively. However, this only holds if $J_{M}^{(1)}=\ldots=J_{M}^{(L)}$, that is, if $\mbox{Ran}(F_{\ell})=I_{M}$ for all $\ell \in \{1,\ldots,L\}$. Conversely, if $\mbox{Ran}(F_{\ell})=I_{M}$ for all $\ell \in \{1,\ldots,L\}$, then the discrete subcopula $D^{\ast}$ has domain $I_{M}^{L}$, and thus we have $D=D^{\ast}$.\qedhere}
\end{enumerate}
\end{proof}

\noindent Theorem \ref{sklarmdc} is especially tailored to and suitable for situations in which dealing with empirical copulas of data with no ties matters. This is for instance the case in the ECC approach proposed by \cite{Schefzik&2013}, which is discussed in the following.

\section{Ensemble copula coupling: An application of multivariate discrete copulas in meteorology}\label{applecc}
\noindent In this section, we relate our concepts to the ensemble copula coupling 
(ECC) approach of \cite{Schefzik&2013}, which is a multivariate statistical 
postprocessing technique for ensemble weather forecasts, and deepen the theoretical considerations in Section 4.2 in \cite{Schefzik&2013}.\\
In state of the art meteorological practice, weather forecasts are derived from ensemble prediction systems, which comprise multiple runs of numerical weather prediction models differing in the initial conditions and/or in details of the parameterized numerical representation of the 
atmosphere \citep{GneitingRaftery2005}. However, ensemble forecasts often reveal biases and dispersion errors. It is thus common that they get statistically postprocessed in order to correct these shortcomigs. Ensemble predictions and their postprocessing lead to probabilistic forecasts in form of predictive probability distributions 
over future weather quantities,
where forecast distributions of good quality are characterized 
by sharpness
subject to calibration \citep{Gneiting2007}. Several 
ensemble postprocessing methods have been proposed, yet many of them, such as Bayesian model averaging (BMA; \cite{Raftery&2005}) or ensemble model output statistics (EMOS; \cite{gnetal2005}), only 
apply to a single weather
quantity at a single location for a single prediction horizon. In 
many applications, however, it is crucial to account for spatial, temporal and 
inter-variable dependence
structures, as in air traffic management or ship routeing, for instance.\\
To address this, ECC as introduced by \cite{Schefzik&2013} offers a simple yet powerful tool, which in a nutshell performs as follows: For each weather variable $i \in \{1,\ldots,I\}$, location $j \in \{1,\ldots,J\}$ and prediction horizon $k \in \{1,\ldots,K\}$ separately, we are given the $M$ forecasts $x_{1}^{\ell},\ldots,x_{M}^{\ell}$ of the original unprocessed raw ensemble, where $\ell:=(i,j,k)$, $\ell \in \{1,\ldots,L\}$, and $L=I \times J \times K$. For each fixed $\ell$, let $\sigma_{\ell}(m):=\operatorname{rk}(x_{m}^{\ell})$ for $m \in \{1,\ldots,M\}$ be the permutation of $\{1,\ldots,M\}$ induced by the order statistics $x_{(1)}^{\ell} \leq \ldots \leq x_{(M)}^{\ell}$ of the raw ensemble, with any ties resolved at random. In a first step, we employ state-of-the-art univariate postprocessing methods such as BMA or EMOS to obtain calibrated and sharp predictive cdfs $F_{X_{\ell}}$, $\ell \in \{1,\ldots,L\}$, for each variable, location and look-ahead time individually. Then, we draw $M$ samples $\tilde{x}_{1}^{\ell},\ldots,\tilde{x}_{M}^{\ell}$ from $F_{X_{\ell}}$ for each $\ell \in \{1,\ldots,L\}$. This can be done, for instance, by taking the equally spaced $\frac{m-\frac{1}{2}}{M}$--quantiles, where $m \in \{1,\ldots,M\}$, of each predictive cdf $F_{X_{\ell}}$, $\ell \in \{1,\ldots,L\}$. In the final ECC step, the $M$ samples (quantiles) are rearranged with respect to the ranks the ensemble members are assigned within the raw ensemble in order to retain the spatial, temporal and inter-variable rank dependence structure and to capture the flow dependence of the raw ensemble.
\begin{figure}[p]
\centering
\subfiguretopcaptrue
\subfigure[Raw Ensemble]{\includegraphics[scale=0.35]{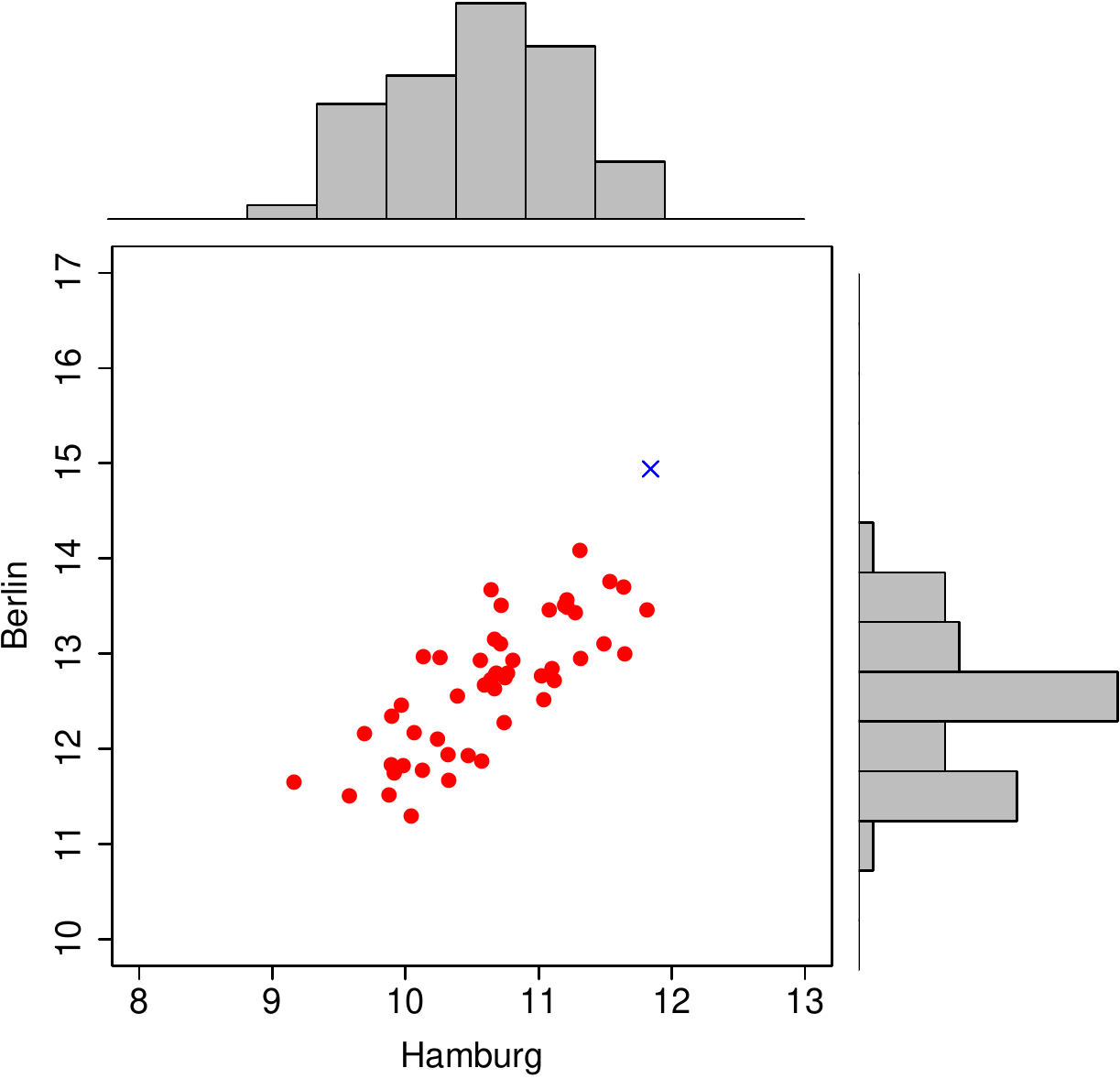}}
\subfigure[Independently Postprocessed Ensemble]{\includegraphics[scale=0.35]{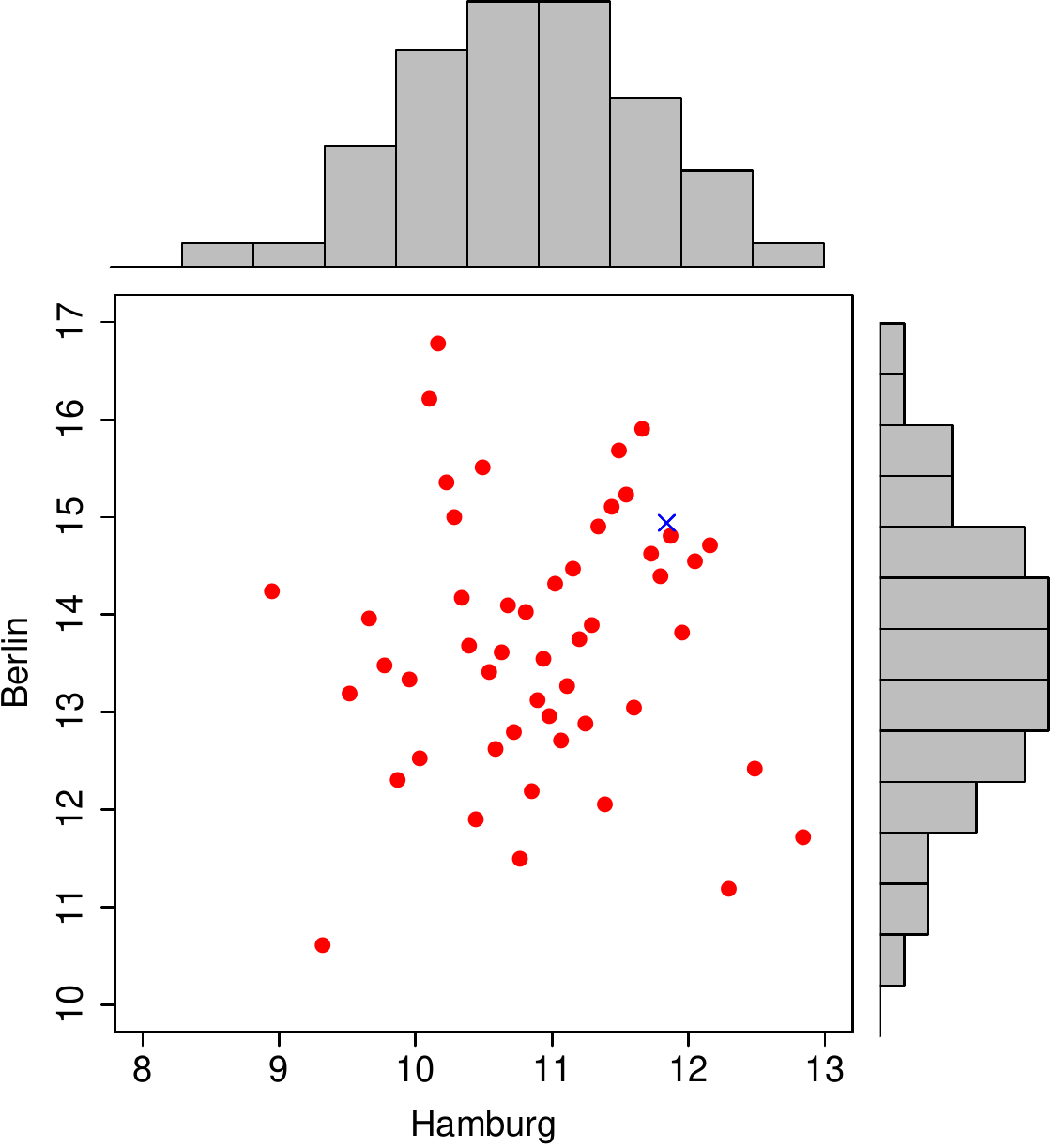}}
\subfigure[ECC Ensemble]{\includegraphics[scale=0.35]{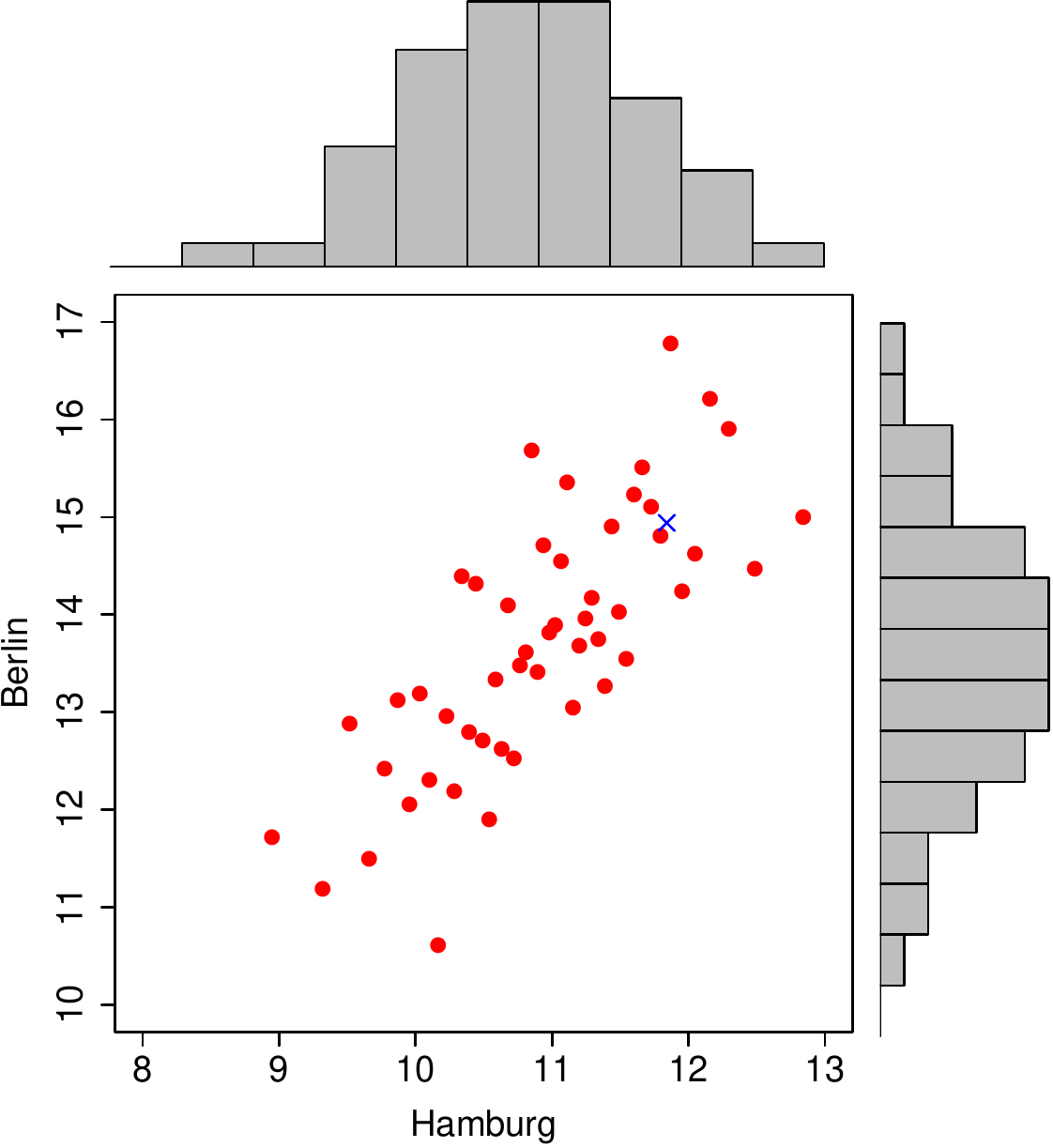}}\newline
\subfigure{\includegraphics[scale=0.35]{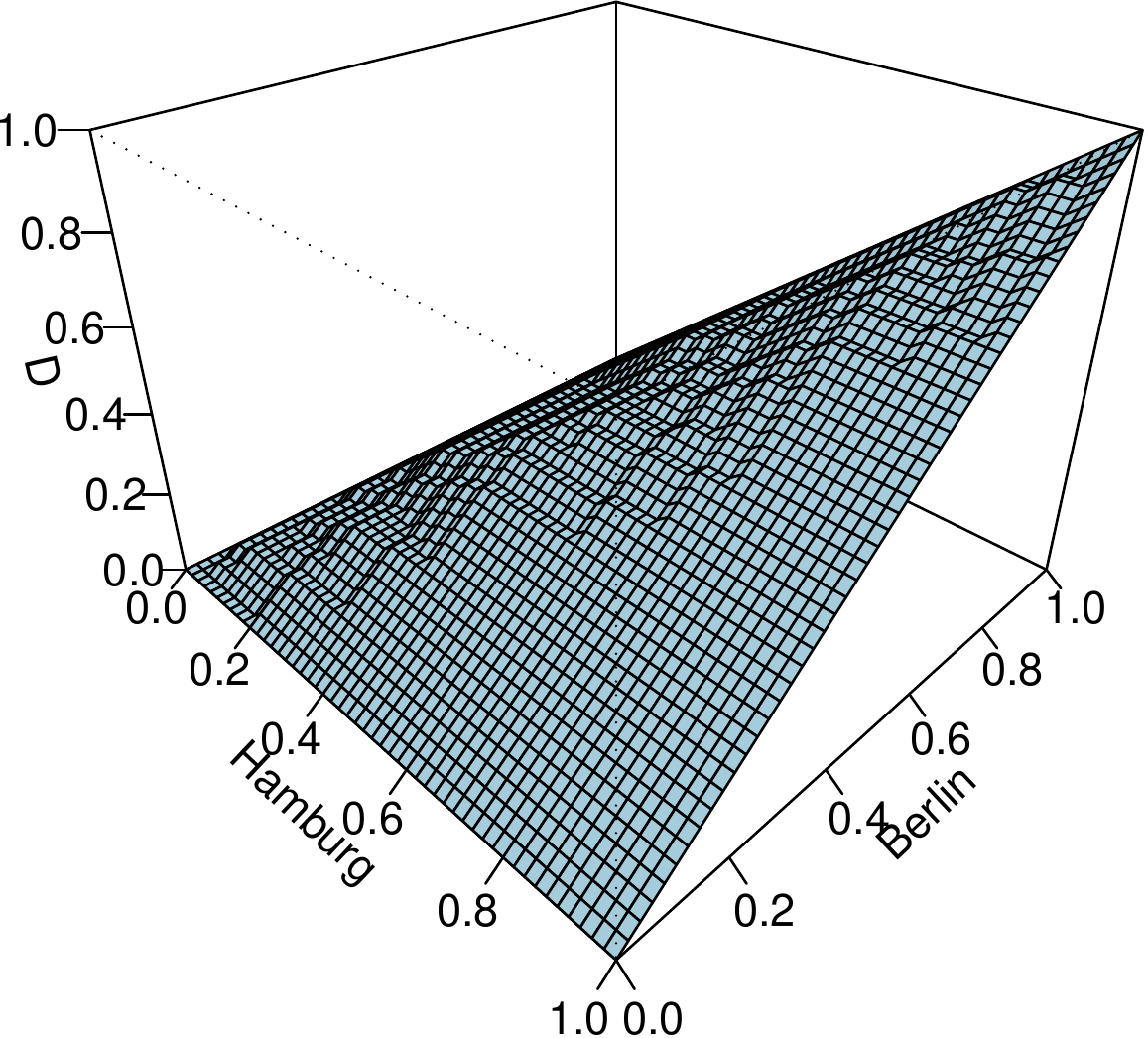}}
\subfigure{\includegraphics[scale=0.35]{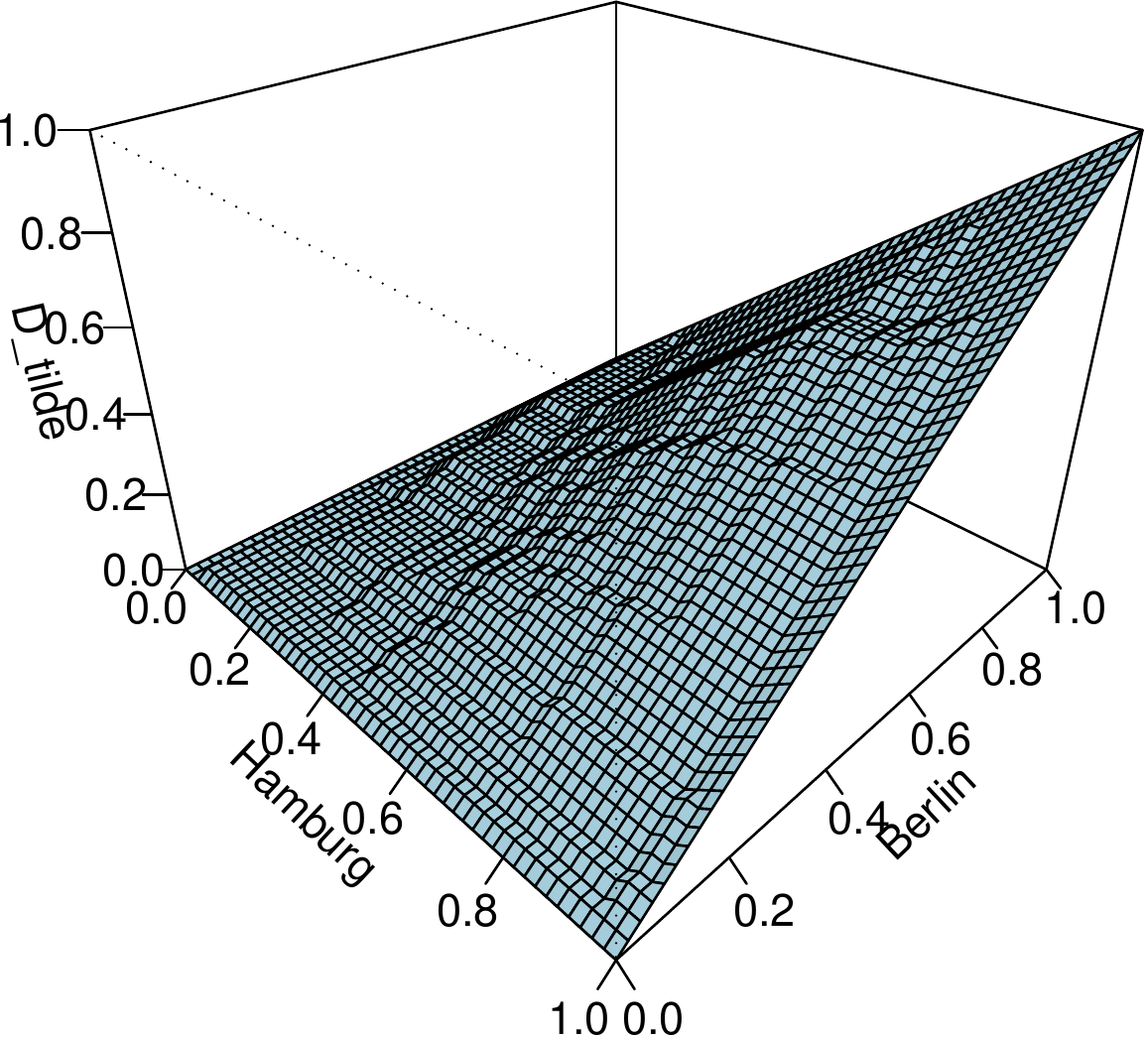}}
\subfigure{\includegraphics[scale=0.35]{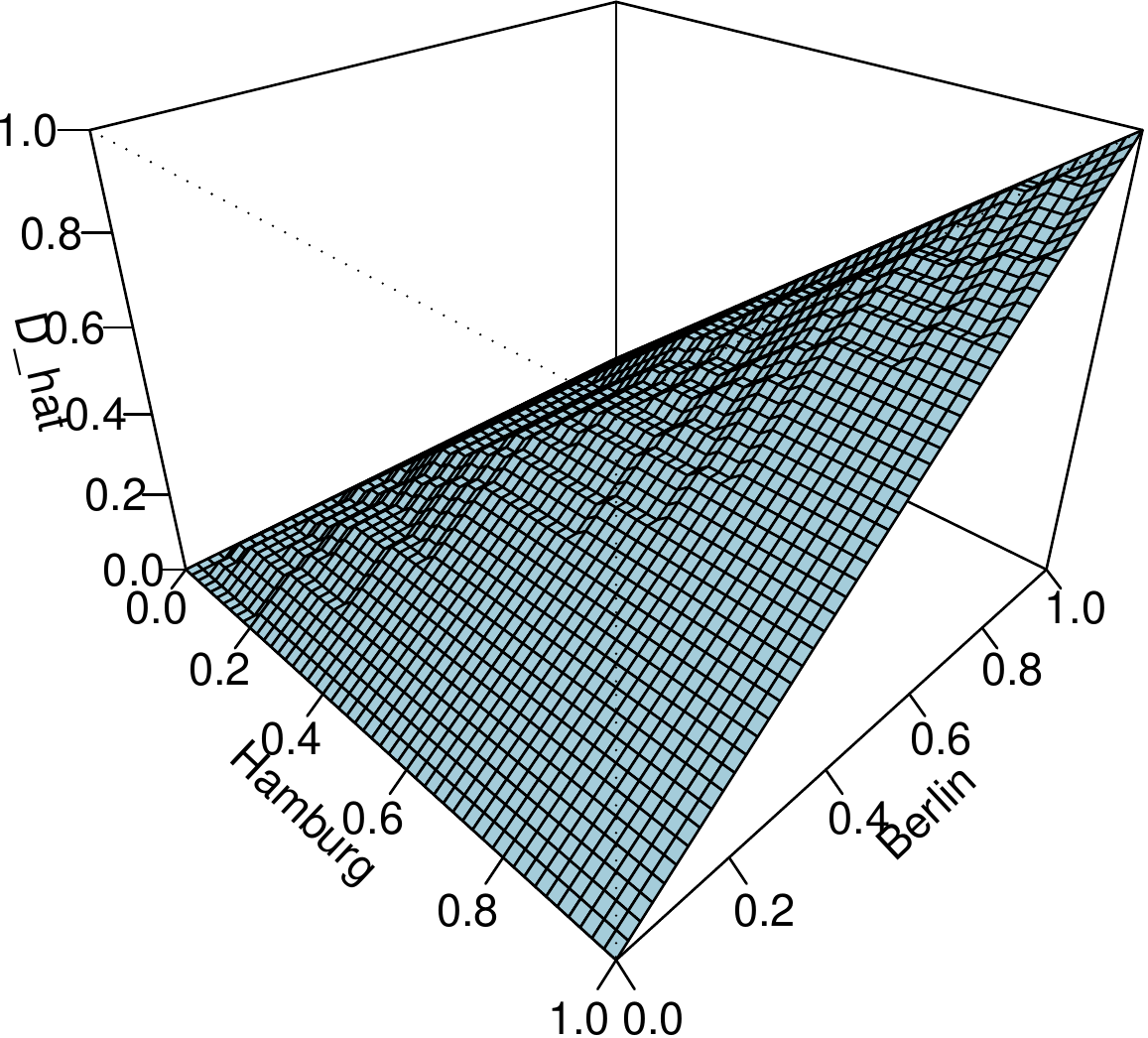}}\newline
\subfigure{\includegraphics[scale=0.35]{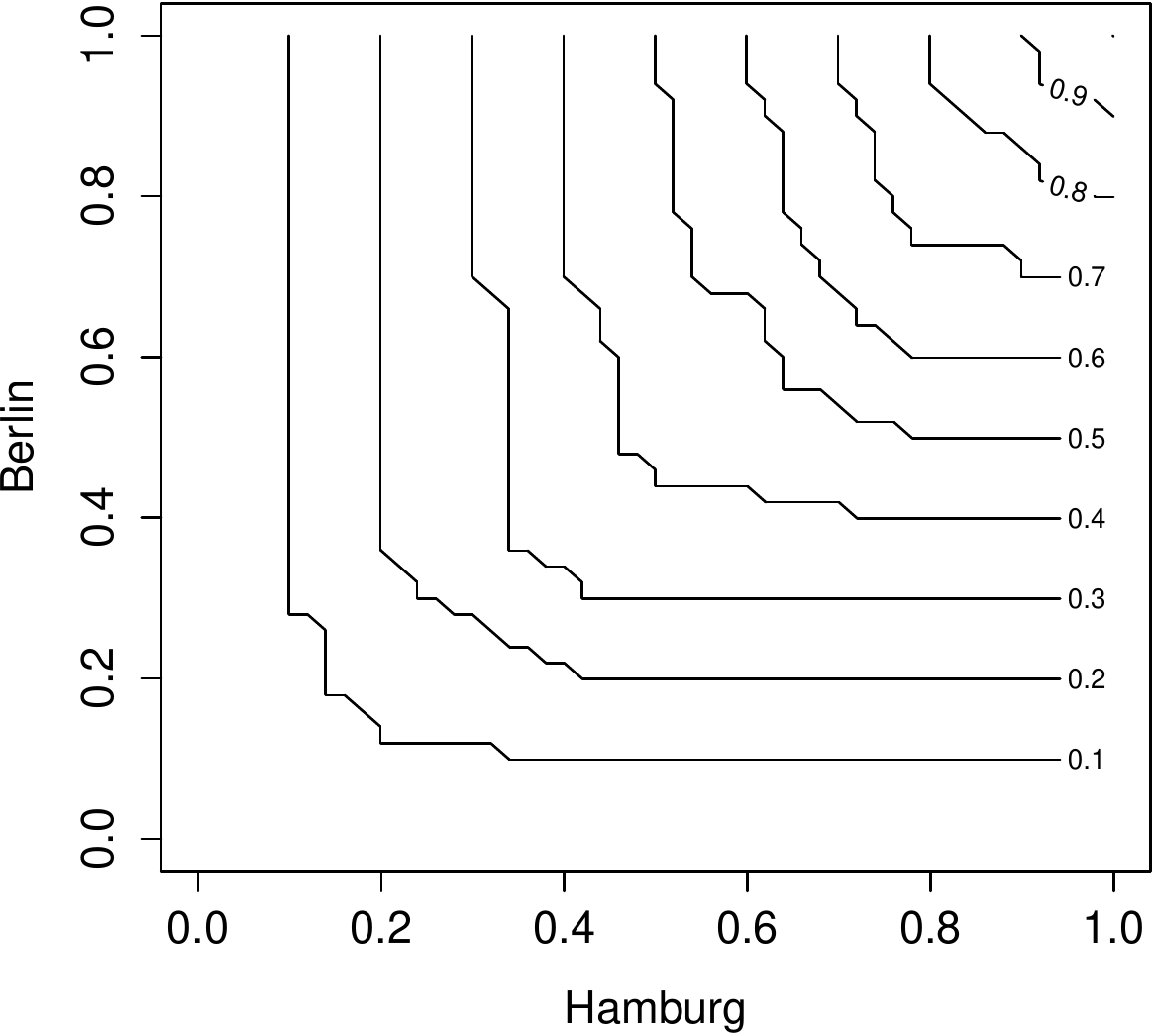}}
\subfigure{\includegraphics[scale=0.35]{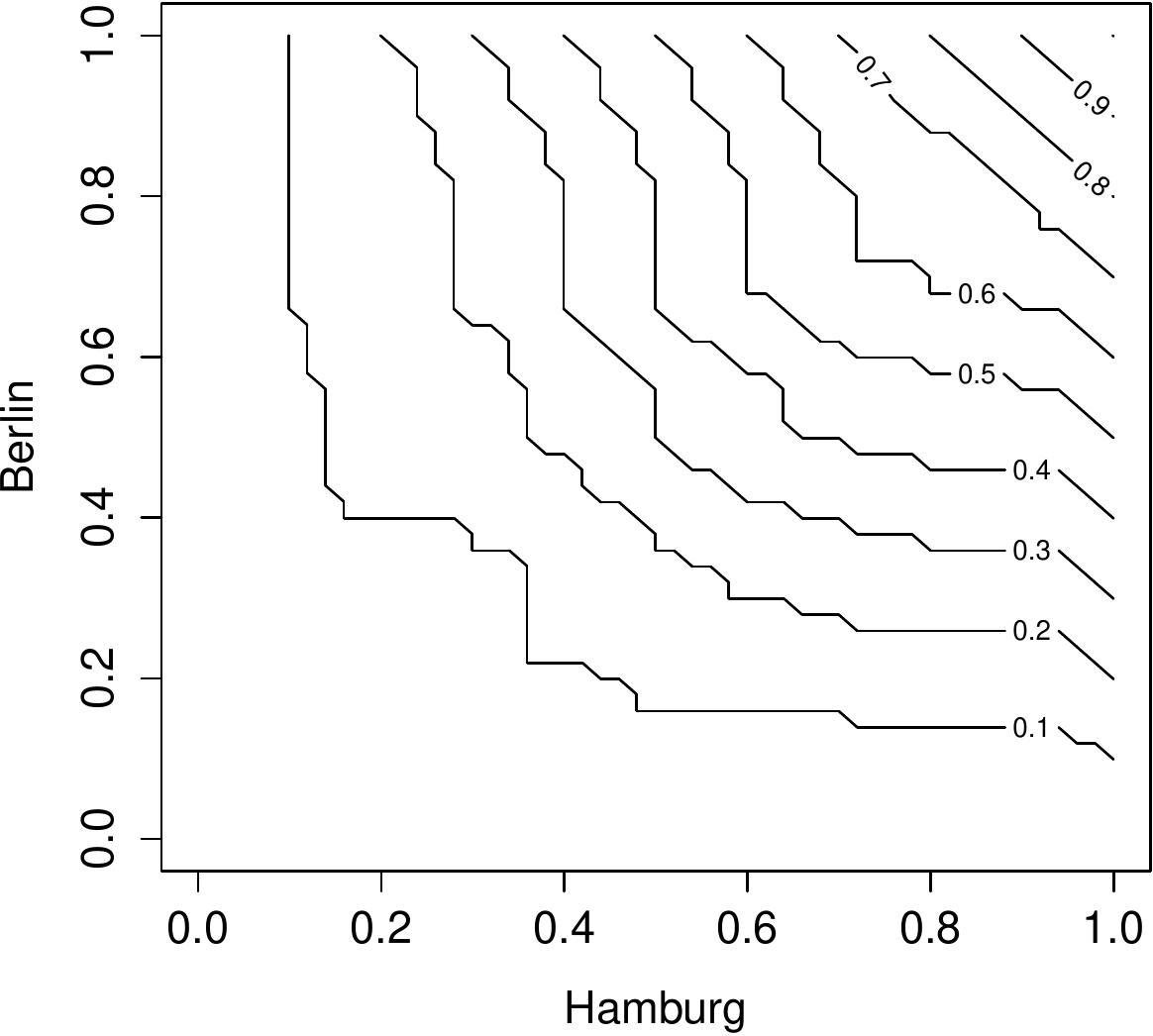}}
\subfigure{\includegraphics[scale=0.35]{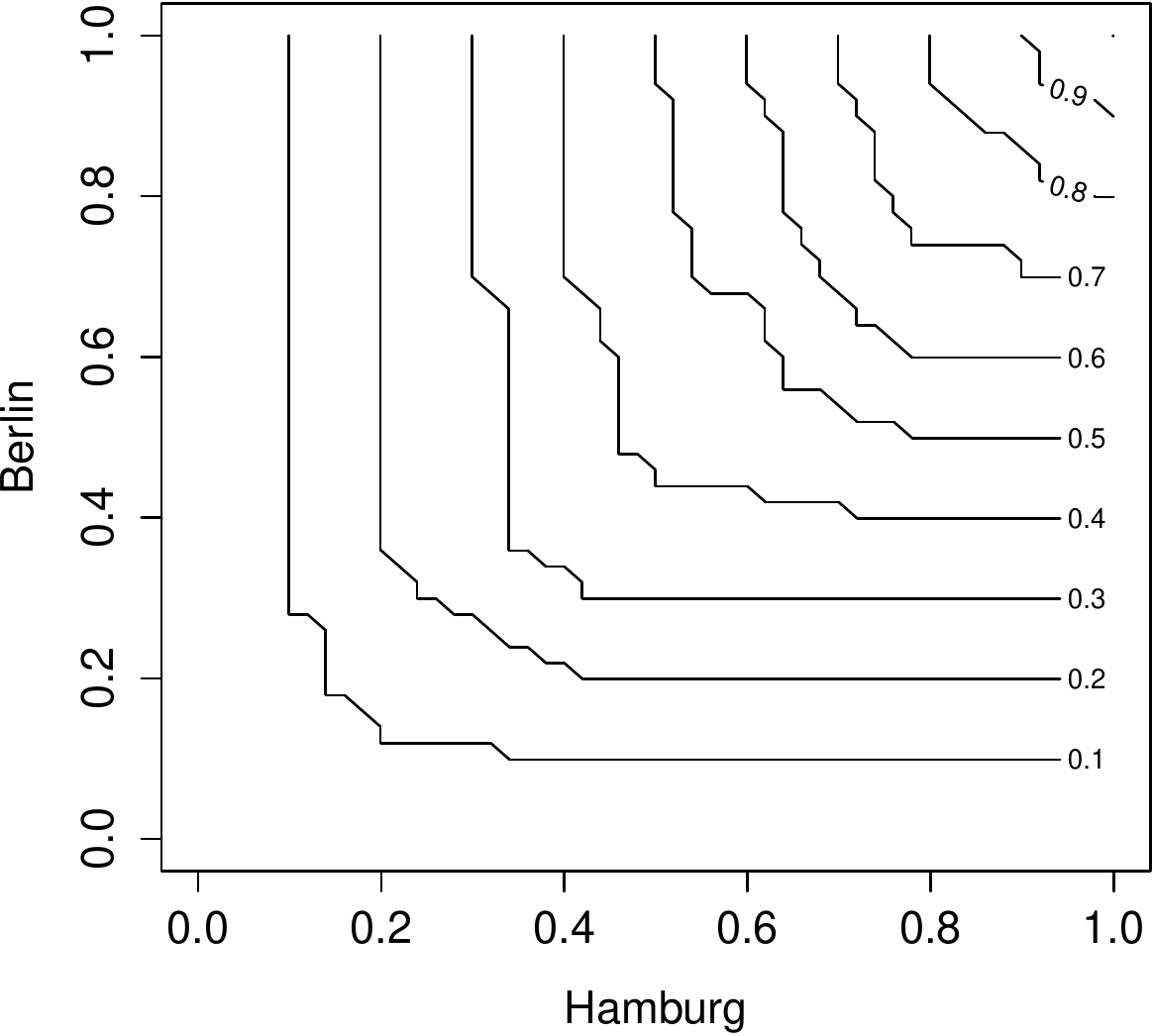}}\newline
\subfigure{\includegraphics[scale=0.35]{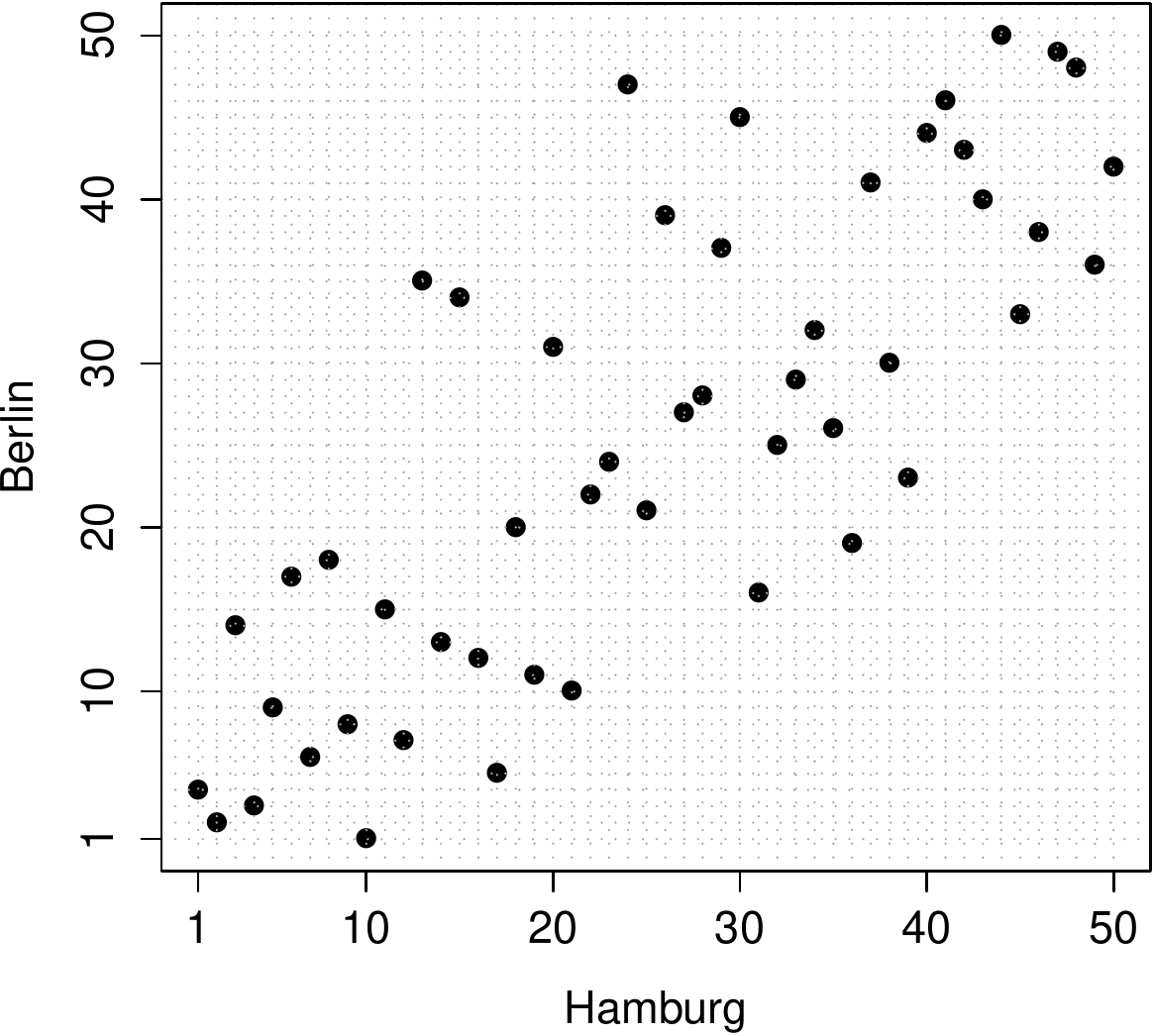}}
\subfigure{\includegraphics[scale=0.35]{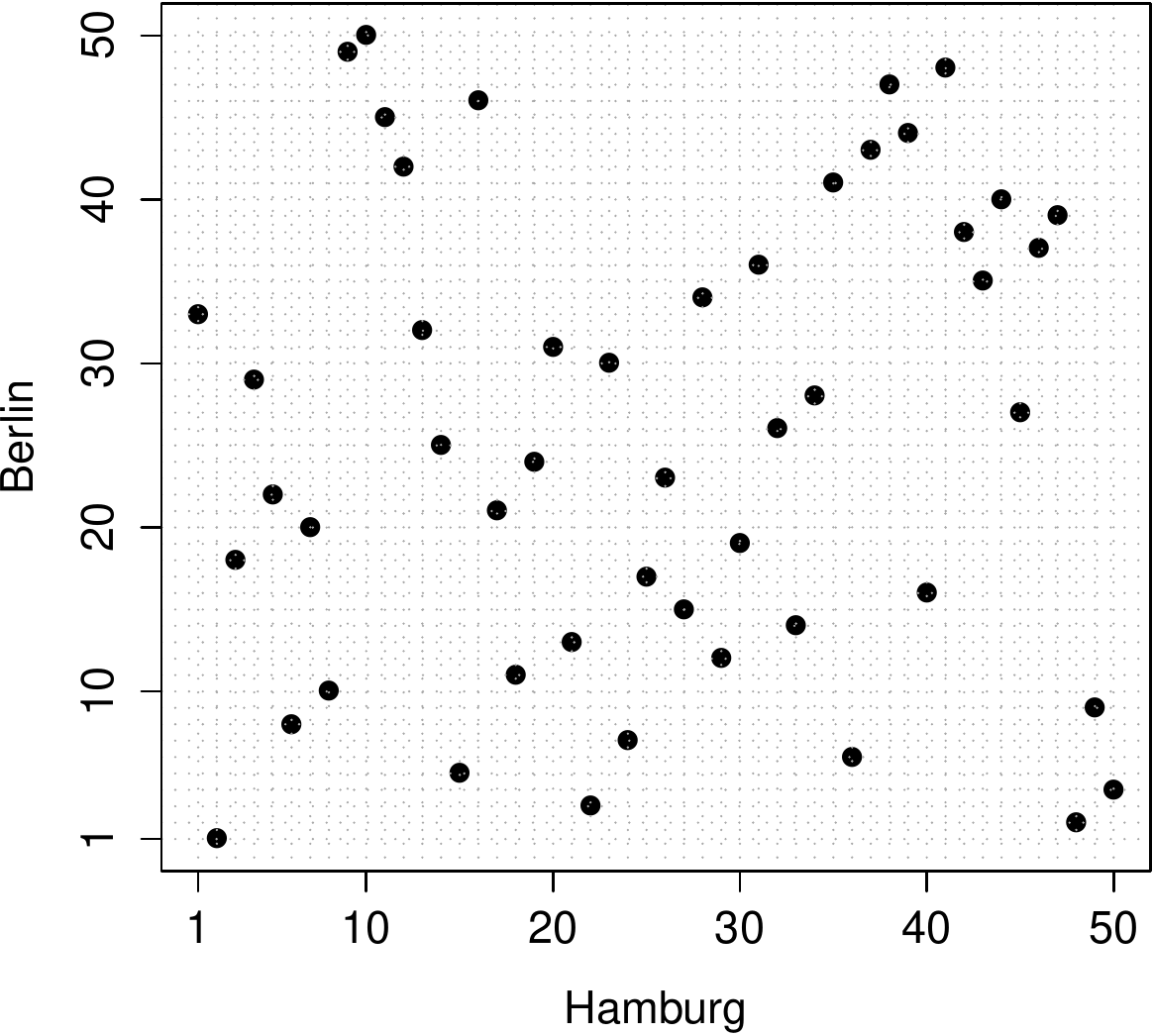}}
\subfigure{\includegraphics[scale=0.35]{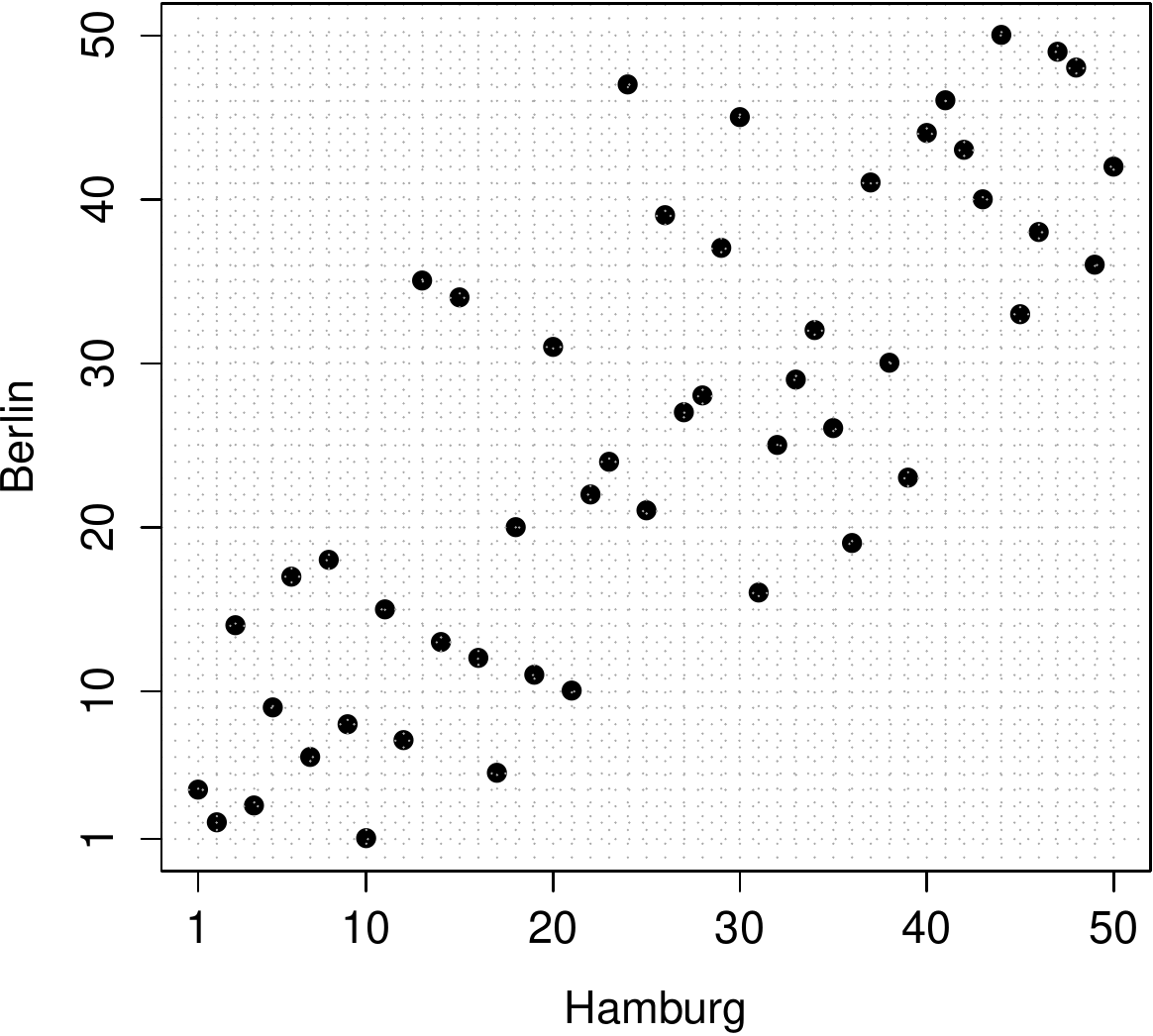}}\newline
\caption{Different ensemble prediction approaches comprising the (a) raw, (b) independently postprocessed and (c) ECC ensemble. First row: Scatterplots with corresponding marginal histograms of 24 hour ahead temperature 
forecasts (in °C) at Berlin and Hamburg, valid on 27 June 2010, 2:00 am local time, based on the 
50-member ECMWF ensemble. The red dots show the forecasts of the corresponding ensemble members, and the 
verifying observation is indicated by the blue cross. Second row: Perspective plots of the corresponding empirical copulas. Third row: Contour plots of the corresponding empirical copulas. Fourth row: Corresponding Latin squares.}
\label{ecc}
\end{figure}
%\begin{figure}[thbp]
%\centering
%\includegraphics[scale=0.7]{figmdcpaper.pdf}
%\caption{Different ensemble prediction approaches. First row: 24 hour ahead temperature 
%forecasts (in °C) at Berlin and Hamburg, valid on 27 June 2010, 2:00 am local time, based on the 
%50-member ECMWF ensemble. The red dots show the forecasts of the corresponding ensemble members, and the 
%verifying observation is indicated by the blue cross.}\label{ecc}
%\end{figure}
That is, the ECC ensemble consists of $\hat{x}_{1}^{\ell}:=\tilde{x}_{(\sigma_{\ell}(1))}^{\ell},\ldots,\hat{x}_{M}^{\ell}:=\tilde{x}_{(\sigma_{\ell}(M))}^{\ell}$ for each $\ell \in \{1,\ldots,L\}$.\\
\cite{Schefzik&2013} and \cite{schuhen} show in several case studies that ECC is well-performing in the sense that the ECC ensemble exhibits better stochastic characteristics than the unprocessed raw ensemble.\\
Although our concepts have been discussed for the general multivariate case, we now consider for illustrative purposes a bivariate ($L=2$) example in the first row of Figure \ref{ecc}, with 24 hour ahead forecasts for temperature at Berlin and Hamburg,  based on the $M=50$-member European Centre for Medium-Range Weather Forecasts (ECMWF) ensemble \citep{molteni,buizza} and valid on 27 June 2010, 2:00 am local time. Univariate postprocessing is performed by BMA. In the left panel of the first row, the unprocessed raw ensemble forecasts are shown, while the plot in the middle presents the independently site-by-site postprocessed ensemble, in which the bivariate rank order characteristics of the unprocessed forecasts from the left pattern are lost, even though biases and dispersion errors have been corrected. Finally, the postprocessed ECC ensemble in the right panel corrects biases and dispersion errors as the independently postprocessed ensemble does, but also takes account of the rank dependence structure given by the raw ensemble.\\
As indicated by its name, ECC has strong connections to copulas, particularly to the notions and results presented before, which is hinted at by \cite{Schefzik&2013} and is investigated in more detail in what follows.\\
To this end, we stick to the notation employed above and let  
$X_{1},\ldots,X_{L}$ be discrete random variables that can take values in
$\{x_{1}^{1},\ldots,x_{M}^{1}\},\ldots,\{x_{1}^{L},\ldots,x_{M}^{L}\}$,  
respectively, where $x_{1}^{\ell},\ldots,x_{M}^{\ell}$ are the $M$  
raw ensemble forecasts for fixed $\ell \in \{1,\ldots,L\}$, that is, for fixed  
weather quantity, location and look-ahead time.  
For convenience, we assume that there are no ties among the  
corresponding values. Considering the multivariate random vector  
$\boldsymbol{X}:=(X_{1},\ldots,X_{L})$, the margins  
$X_{1},\ldots,X_{L}$ are uniformly distributed with step $\frac{1}{M}$,  
and their corresponding univariate cdfs  
$F_{X_1},\ldots,F_{X_L}$ hence take values in $I_{M}$, that is,  
$\mbox{Ran}(F_{X_{1}})=\ldots=\mbox{Ran}(F_{X_{L}})=I_{M}$. Moreover, we have $\mbox{Ran}(H)=I_{M}$ for the multivariate cdf $H$ of $\boldsymbol{X}$. According to the multivariate discrete  
version of Sklar's theorem tailored to the ECC framework here, compare Theorem \ref{sklarmdc}, there exists a uniquely  
determined irreducible discrete copula $D:I_{M}^{L} \rightarrow I_{M}$
such that
\begin{equation*}
H(u_{1},\ldots,u_{L}) = D(F_{X_1}(u_{1}),\ldots,F_{X_L}(u_{L})) \mbox{\,\,for\,\,} (u_{1},\ldots,u_{L}) \in \overline{\mathbb{R}}^{L},
\end{equation*}
that is, the multivariate distribution is connected to its univariate
margins via $D$.\\
Following and generalizing the statistical interpretation of discrete  
copulas for the bivariate case by \cite{Mesiar2005},
\begin{equation*}
D\left(\frac{i_{1}}{M},\ldots,\frac{i_{L}}{M}\right)=\mathbb{P}(H \in  
[-\infty,u_{1}]  \times \cdots \times [-\infty,u_{L}]),
\end{equation*}
where $u_{1},\ldots,u_{L} \in \overline{\mathbb{R}}$ such that $F_{X_1}(u_{1})=\mathbb{P}(X_{1} \leq u_{1})=\frac{i_{1}}{M},\ldots,F_{X_L}(u_{L})=\mathbb{P}(X_{L} \leq u_{L})=\frac{i_{L}}{M}$, that is, $D\left(\frac{i_{1}}{M},\ldots,\frac{i_{L}}{M}\right)=\mathbb{P}(X_{1} \leq u_{1},\ldots,X_{L} \leq u_{L})$. To describe the discrete probability distribution of the random vector $\boldsymbol{X}$, we set $\alpha_{i_{1} \ldots i_{L}} := \mathbb{P}(X_{1}=x_{(i_{1})}^{1},\ldots,X_{L}=x_{(i_{L})}^{L})$, where $x_{(i_{\ell})}^{\ell}$, $\ell \in \{1,\ldots,L\}$, denote the corresponding order statistics from samples describing the values of $X_{1},\ldots,X_{L}$. Obviously, $\alpha_{i_{1} \ldots i_{L}} \in \left\{0,\frac{1}{M}\right\}$ for all $i_{1},\ldots,i_{L} \in \{1,\ldots,M\}$. Hence, $a_{i_{1} \ldots i_{L}}:=M \alpha_{i_{1} \ldots i_{L}} \in \{0,1\}$ for $i_{1},\ldots,i_{L} \in \{1,\ldots,M\}$, $A:=(a_{i_{1} \ldots i_{L}})_{i_{1},\ldots,i_{L}=1}^{M}$ is a permutation array, and
\begin{equation*}
\frac{1}{M} \sum\limits_{j_{1}=1}^{i_{1}} \cdots \sum\limits_{j_{L}=1}^{i_{L}} a_{j_{1} \ldots j_{L}} = D \left(\frac{i_{1}}{M},\ldots,\frac{i_{L}}{M}\right),
\end{equation*}
which is in accordance to Theorem \ref{equiv}.
\\
Analogously, the same considerations hold for both the independently  
postprocessed ensemble consisting of the $M$ samples  
$\tilde{x}_{1}^{\ell},\ldots,\tilde{x}_{M}^{\ell}$, $\ell \in \{1,\ldots,L\}$,  
and the ECC ensemble $\hat{x}_{1}^{\ell},\ldots,\hat{x}_{M}^{\ell}$,  
$\ell \in \{1,\ldots,L\}$, that is,
\begin{equation*}
\tilde{H}(u_{1},\ldots,u_{L}) =  
\tilde{D}(F_{\tilde{X}_{1}}(u_{1}),\ldots,F_{\tilde{X}_{L}}(u_{L}))  
\mbox{\,\,\,\,for\,\,\,\,} (u_{1},\ldots,u_{L}) \in \overline{\mathbb{R}}^{L}
\end{equation*}
and
\begin{equation*}
\hat{H}(u_{1},\ldots,u_{L}) =  
\hat{D}(F_{\hat{X}_{1}}(u_{1}),\ldots,F_{\hat{X}_{L}}(u_{L})) \mbox{\,\,\,\,for\,\,\,\,} (u_{1},\ldots,u_{L}) \in \overline{\mathbb{R}}^{L},
\end{equation*}
in obvious notation. Although both the independently postprocessed and the ECC ensemble  
have the same marginal distributions, that is,  
$F_{\tilde{X}_{1}}=F_{\hat{X}_{1}},\ldots,F_{\tilde{X}_{L}}=F_{\hat{X}_{L}}$, as is illustrated by the marginal histograms in the first row of our example in Figure \ref{ecc}, they differ drastically in their multivariate rank dependence  
structure. Since the ECC ensemble is designed in  
the manner that it inherits the rank dependence pattern from the raw  
ensemble, the considerations above yield that $D=\hat{D}$. Thus, the raw and the ECC ensemble are associated with the same irreducible  
multivariate discrete copula modeling the dependence. This is visualized in the second and third row of Figure \ref{ecc}, where the perspective and contour plots, respectively, of the empirical copulas linked to the different ensembles in our illustrative example are shown, both suitably indicating rather high dependence. On the other hand, the perspective and contour plots of the empirical copula associated with the independently postprocessed ensemble in the mid-panel of Figure \ref{ecc} are not far away from those of the independence copula $\Pi$ introduced in Section \ref{basics}. According to the equivalences in Section \ref{dcstar}, the raw and the ECC ensemble are also related to the same Latin square of order $M=50$, as can be seen in the fourth row in Figure \ref{ecc}.\\
Hence, ECC indeed can be  
considered as a copula approach, as it comes up with a postprocessed, discrete $L$-dimensional distribution, which is by Theorem \ref{sklarmdc} constructed from the $L$ univariate predictive cdfs $F_{X_1},\ldots,F_{X_L}$ obtained by the postprocessing and the empirical copula $D$ induced by the unprocessed raw ensemble. Conversely, each multivariate distribution with fixed univariate margins yields a uniquely determined empirical copula $D$, which defines the rank dependence structure in our setting.\\
Although several multivariate copula-based methods for discrete data have been proposed, for example recently by \cite{Panagiotelis&2012} using vine and pair copulas, we feel that our discrete copula approach still provides an appropriate and useful alternative to these methods. ECC is especially valuable when being faced with extremely high-dimensional data, as is the case in weather forecasting, where one has to deal with several millions of variables. Since its crucial reordering step is computationally non-expensive, one of the major advantages of ECC is that it practically comes for free, once the univariate postprocessing is done. However, BMA and EMOS as univariate postprocessing methods are already implemented efficiently in the R packages \texttt{ensembleBMA} and \texttt{ensembleMOS}, respectively, which are freely available at \url{http://cran.r-project.org}. Hence, with the discrete copula-based non-parametric ECC approach, we can circumvent the problems that arise when using parametric methods, such as computational unfeasibility. In addition, ECC offers a simple and intuitive, yet powerful technique that goes without complex modeling or sophisticated parameter fitting in multivariate copula models, which work well in comparably low dimensional settings \citep{moeller,SchoelzelFriederichs2008}, but tend to fail in very high dimensions. The notion of discrete copulas arises naturally in the context of the ECC approach. Furthermore, as documented in Section 4.4 in \cite{Schefzik&2013}, the discrete copula notion presented in the paper at hand can be interpreted as an overarching concept and theoretical frame not only for ECC, but also for other ensemble postprocessing methods that have recently appeared in the meteorological literature, and applies in other settings as well.

\section*{Acknowledgments}
\noindent This work has been supported by the Volkswagen Foundation under the 
"Mesoscale Weather Extremes: Theory, Spatial Modeling and Prediction 
(WEX-MOP)" project, which is gratefully acknowledged. Moreover, 
the author thanks Tilmann Gneiting, Thordis 
Thorarinsdottir and two anonymous reviewers of an earlier version of the paper for providing helpful comments, hints and suggestions 
in the course of the development of the work at hand.

\bibliographystyle{plainnat}
\bibliography{biblioecc}

\end{document}